\documentclass[journal,onecolumn]{IEEEtran}
\IEEEoverridecommandlockouts
\usepackage{bm}
\usepackage{amsmath,amssymb,amsfonts,amsthm}
\usepackage{algpseudocode}
\usepackage{algorithm}
\usepackage{physics}
\usepackage[caption=false]{subfig}

\renewenvironment{proof}{{\bfseries Proof.}}{\qedsymbol}
\renewcommand\qedsymbol{$\blacksquare$}
\theoremstyle{definition}

\newtheorem{theoremx}{Theorem}

\usepackage{tikz}

\newcommand{\tp}{^\top}
\DeclareDocumentCommand\diag{}{\opbraces{\operatorname{diag}}}
\DeclareDocumentCommand\rank{}{\opbraces{\operatorname{rank}}}
\DeclareDocumentCommand\tr{}{\opbraces{\operatorname{tr}}}
\DeclareDocumentCommand\vec{}{\opbraces{\operatorname{vec}}}
\def\BibTeX{{\rm B\kern-.05em{\sc i\kern-.025em b}\kern-.08em
    T\kern-.1667em\lower.7ex\hbox{E}\kern-.125emX}}
\begin{document}

\title{Fast Multiagent Formation Stabilization with\\Sparse Universally Rigid Frameworks
\thanks{This work is partially funded by the Sensor AI Lab, under the AI Labs program of Delft University of Technology, and by the RVO (Missie gedreven Onderzoek, Ontwikkeling en Innovatie (MOOI)) Aerial Quadcopter Units for Aquatic Flow Investigation and Nautical Data (AQUAFIND) project.}
}

\author{\IEEEauthorblockN{Zhonggang Li, Geert Leus, and Raj Thilak Rajan}\\
\IEEEauthorblockA{\textit{Signal Processing Systems, Department of Microelectronics, Faculty of EEMCS} \\
\textit{Delft University of Technology}\\
Delft, The Netherlands \\
\{z.li-22, g.j.t.leus, r.t.rajan\}@tudelft.nl}
}

\maketitle

\begin{abstract}
Affine formation control (AFC) is a distributed networked control system that has recently received increasing attention in various applications. AFC is typically achieved using a generalized consensus system where the stress matrix, which encodes the graph structure, is used instead of a graph Laplacian. Universally rigid frameworks (URFs) guarantee the existence of the stress matrix and have thus become the guideline for such a network design. In this work, we propose a convex optimization framework to design the stress matrix for AFC without predefining a rigid graph. We aim to find a resulting network with a reduced number of communication links, but still with a fast convergence speed. We show through simulations that our proposed solutions can yield a more sparse graph, while admitting a faster convergence compared to the state-of-the-art solutions.

\end{abstract}

\begin{IEEEkeywords}
network design, consensus, graph rigidity, formation control, multiagent systems, stress matrix
\end{IEEEkeywords}

\section{Introduction}
Distributed formation control is an essential task for robotic swarming applications \cite{oh2015survey, li2023geometry, marel2022distributed}, where agents use relative information such as interagent distances, relative positions, etc., to achieve and maintain a desired geometric pattern in two- or three-dimensional space. Such a system is typically characterized by a framework \cite{alfakih2013affine} consisting of (a) a configuration, which represents the collection of agent positions forming a geometric pattern, and (b) a graph where the edges denote the communication links for information exchange. Such a networked system shares similarities with sensor networks, relative localization, distributed optimization, etc., where the communication pattern and capacity play a critical role \cite{heusdens2024distributed, rajan2019relative, zheng2009wireless}.

Recently, affine formation control (AFC) has gained increasing attention due to its flexibility in maneuvering while maintaining coordinated motion \cite{zhao2018affine}. Unlike traditional formation control that enforces rigid geometric constraints, AFC allows for controlled scaling, shearing, etc., making it particularly useful for dynamic environments and adaptive mission planning. AFC follows a consensus-like framework, where control inputs are derived as linear combinations of position differences between neighboring agents. As compared to the conventional average consensus using the graph Laplacian, AFC adopts a stress matrix - a generalization of the Laplacian with weighted edges. The universal rigidity of the associated framework guarantees the existence of stress with specific properties that enable AFC. Universally rigid frameworks (URFs) are also applied in network localization \cite{zhu2010universal}, tensegrity frameworks \cite{connelly2022frameworks}, etc.

Interagent communication is essential in distributed formation control, making the reduction of communication load a crucial necessity, due to bandwidth and power constraints of agents. However, a sparse network with few edges is typically less efficient in information exchange, and thus slows down the convergence of consensus-based algorithms including AFC. This has been an ongoing discussion in consensus theory involving the algebraic connectivity \cite{kim2005maximizing, olfati2007consensus}. In AFC, designing the stress matrix with fewer edges while preserving rigidity and convergence speed is challenging. Earlier mathematical literature discusses URF design using special geometries \cite{kelly2014class, connelly2022frameworks}, later followed by numerical stress calculations for AFC applications \cite{lin2015necessary, zhao2018affine}. However, manual graph design becomes impractical for larger networks, and structured geometries like Grünbaum polygons \cite{kelly2014class} can create unbalanced networks that hinder robustness. Thus, generative methods \cite{yang2018constructing, xiao2022framework} that directly design a possibly sparse stress matrix from a given configuration without predefined graphs are preferred. A sparse nullspace reconstruction approach \cite{yang2018constructing} achieves substantial sparsification but fails to maintain convergence from an arbitrary initial configuration in many cases. More recently, mixed-integer semidefinite programming (MISDP) is adopted in \cite{xiao2022framework} for sparse stress design with optimized convergence speed, but raises concerns about optimality and computational complexity.

In this paper, we present an efficient convex optimization framework to design the stress matrix with two key objectives, (a) network sparsification and (b) convergence acceleration, and we further give insights on how to balance these two rather conflicting objectives. In Section \ref{sec: preliminaries}, we introduce the preliminaries of graph rigidity and how stress matrices play a critical role in AFC. We provide a preliminary problem formulation based on the properties of stress and our optimization objective, which turns out to be a non-convex problem. In Section \ref{sec: methodology}, we relax the non-convex problem into a concise and tractable convex problem with insights into the choice of hyperparameters. Finally, we compare our proposed solution with the aforementioned state-of-art using several numerical examples in Section \ref{sec: simulations}, before reaching a conclusion and giving potential future research directions in Section \ref{sec: conclusions}.

\textit{Notations.} Vectors and matrices are represented by lowercase and uppercase boldface letters, respectively such as $\bm{a}$ and $\bm{A}$. Sets and graphs are represented using calligraphic letters, e.g., $\mathcal{A}$. Vectors of length $N$ of all ones and zeros are denoted by $\bm{1}_N$ and $\bm{0}_N$, respectively, and their matrix versions are similarly $\bm{1}_{M\times N}$ and $\bm{0}_{M\times N}$. An identity matrix of size $N$ is denoted by $\bm{I}_N$. The Kronecker product is denoted by $\otimes$, which has relationships with the vectorization operator $\vec(\cdot)$ \cite{petersen2008matrix}. The $\diag(\cdot)$ operator creates a diagonal matrix from a vector and $\tr(\cdot)$ denotes the trace operator. $\lambda_k(\bm{A})$ denotes the $k$-th smallest eigenvalue of a symmetric matrix $\bm{A}$.

\section{Fundamentals}\label{sec: preliminaries}

\subsection{Graphs and Rigidity Theory}
Consider $N$ mobile agents moving in $D$-dimensional Euclidean space where $N\geq D+1$. An undirected graph $\mathcal{G} = (\mathcal{V}, \mathcal{E})$ is used to model a communication network, where the vertices $\mathcal{V} = \qty{1,...,N}$ denote the agents, and the edges $\mathcal{E}\subseteq\mathcal{V}\times \mathcal{V}$ denote the pairwise interactions, e.g., information exchange. We use $N = \qty|\mathcal{V}|$ and $M =\qty|\mathcal{E}|$ as a shorthand notation for the number of vertices and edges, respectively. The set of neighbors of a node $i$ is defined as $\mathcal{N}_i = \{j\in\mathcal{V}: (i,j)\in\mathcal{E}\}$. Let $\bm{p}_i\in\mathbb{R}^D$ be the position of node $i\in\mathcal{V}$ and the collection of all nodes, the \textit{configuration}, is $\bm{P} = \qty[\bm{p}_1,...,\bm{p}_N]\in\mathbb{R}^{D\times N}$. A \textit{generic configuration} \cite{connelly2005generic} has algebraically independent node coordinates, i.e., no geometric constraints among nodes. Classic nongeneric configurations include nodes forming a line, a regular polygon, etc. 

A \textit{framework} $\mathcal{F} = (\mathcal{G},\bm{P})$ is a tuple of the graph and its associated configuration. Intuitively, the rigidity of frameworks can be judged by whether a different configuration exists given the distances between the nodes on the edges of the graph. More formally, from \cite{lin2015necessary}, let $\bm{P}' = \qty[\bm{p}'_1,...,\bm{p}'_N]\in\mathbb{R}^{D\times N}$, then two frameworks $(\mathcal{G},\bm{P})$ and $(\mathcal{G},\bm{P}')$ are equivalent if $\forall \qty(i,j)\in\mathcal{E}, \norm{\bm{p}_i-\bm{p}_j}_2 = \norm{\bm{p}'_i - \bm{p}'_j}_2$, and they are congruent if $\forall i,j\in\mathcal{V}, \norm{\bm{p}_i-\bm{p}_j}_2 = \norm{\bm{p}'_i - \bm{p}'_j}_2$. The global rigidity of a framework $\mathcal{F}$ defined in $\mathbb{R}^D$ requires that all frameworks that are equivalent to $\mathcal{F}$ are also congruent to it. Note that global rigidity is sometimes simply referred to as rigidity. Frameworks that are not globally rigid are \textit{flexible}. The universal rigidity of a framework $\mathcal{F}$ defined in $\mathbb{R}^D$ requires $\mathcal{F}$ to be globally rigid in $\mathbb{R}^{D'}$ for any $D'\geq D$. Fig. \ref{fig: urf examples} presents a few examples to illustrate the rigidities. 

Universal rigidity is algebraically represented by an \textit{equilibrium stress} $\omega_{ij}\in\mathbb{R}$ for every edge $(i,j) \in \mathcal{E}$, leading to a set of weights associated with the edges, which satisfy
\begin{equation}\label{equ: stress def}
    \sum_{(i,j)\in\mathcal{E}}\omega_{ij}(\bm{p}_i-\bm{p}_j) = \bm{0}_D.
\end{equation}
A more compact form of (\ref{equ: stress def}) is $\bm{\Omega}\bm{P}\tp=\bm{0}_{N\times D}$, where $\bm{\Omega}\in\mathbb{R}^{N\times N}$ is called the \textit{stress matrix} and is defined as
\begin{equation}
    [\bm{\Omega}]_{ij}=\begin{cases}
  0,& \text{if } (i,j)\notin\mathcal{E}\\
  -\omega_{ij},& \text{if } i\neq j, (i,j)\in\mathcal{E}\\
  \sum_{k\in\mathcal{N}_i}\omega_{ik}&  \text{if } i=j
\end{cases}.
\end{equation}
Alternatively, the stress matrix can be defined using the graph incidence matrix $\bm{B}\in\mathbb{R}^{N\times M}$ as 
\begin{equation}\label{equ: stress inc}
    \bm{\Omega} = \bm{B}\diag(\bm{\omega})\bm{B}\tp,
\end{equation}
where $\bm{\omega}\in\mathbb{R}^M$ is a vector containing all the equilibrium stresses. Note that $\bm{\Omega}$ reduces to a standard graph Laplacian if $\diag(\bm{\omega})=\bm{I}_M$, i.e., equal weights for the edges. It also has $\bm{1}_N$ in the nullspace like the Laplacian. The following theorem establishes the important properties of $\bm{\Omega}$ related to universal rigidity.
\begin{theoremx}\label{lmm: alt improper rot}
(\textit{Universally Rigid Frameworks and Stress Matrices}) Given a framework $\mathcal{F} = (\mathcal{G},\bm{P})$ with $\bm{P}$ being a generic configuration, $\mathcal{F}$ is universally rigid if and only if there exists a positive semidefinite stress matrix $\bm{\Omega}$ with rank $N-D-1$.
\end{theoremx}
\begin{proof}
    See \cite{gortler2014characterizing}, \cite{zhao2018affine}, \cite{xiao2022framework}.
\end{proof}
\noindent 

\noindent Although a generic configuration is generally assumed for mathematical guarantees, nongeneric geometries are more interesting for formation control applications. It is worth mentioning that our stress design approach works for both generic and nongeneric configurations.

\begin{figure}[t]
	\centering	

        \subfloat[\scriptsize flexible]{\raisebox{0ex}
		{\includegraphics[width=0.13\textwidth]{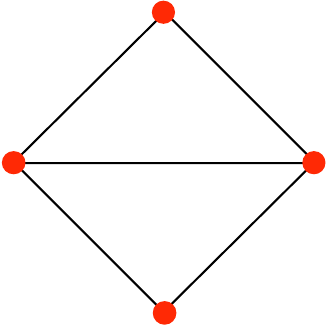}}%
	}
        \hspace{3mm}
	\subfloat[\scriptsize globally rigid]{\raisebox{0ex}
		{\includegraphics[width=0.13\textwidth]{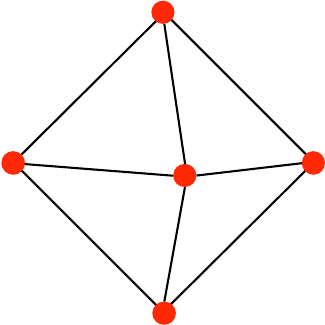}}%
	}
    \hspace{3mm}
        \subfloat[\scriptsize universally rigid]{\raisebox{0ex}
    		{\includegraphics[width=0.13\textwidth]{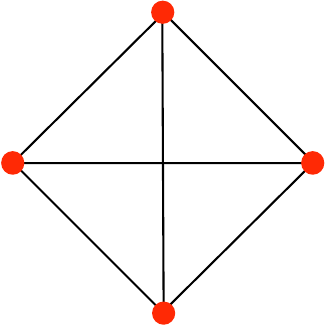}}%
    	}

	\caption{\small Examples of frameworks in $\mathbb{R}^2$ with increasing rigidity \cite{lin2015necessary}}
	\label{fig: urf examples}
		
\end{figure}

\subsection{Consensus-Based Formation Control}
Let $\bm{z} = \qty[\bm{z}_1\tp,...,\bm{z}_N\tp]\tp\in\mathbb{R}^{DN}$, where $\bm{z}_i\in\mathbb{R}^D, \forall i\in\mathcal{V}$, be the true positions of agents that are governed by the single-integrator dynamics $\dot{\bm{z}}_i = \bm{u}_i$ where $\bm{u}_i$ is a velocity control input to be computed by agent $i$. Furthermore, recall the stress matrix from (\ref{equ: stress def}) that satisfies $\bm{\Omega}\bm{P}\tp=\bm{0}$, which can also be written as 
\begin{equation}\label{equ: equilibrium config}
\qty(\bm{\Omega}\otimes\bm{I}_D)\bm{p}=\bm{0},   
\end{equation}
where $\bm{p} = \vec(\bm{P})$ using the Kronecker-vectorization property. As such, a consensus-like control law \cite{lin2015necessary, zhao2018affine} can be designed as $\bm{u}_i = -\sum_{j\in\mathcal{N}_i}\omega_{ij}(\bm{z}_i-\bm{z}_j), \forall i\in\mathcal{V}$, such that the collective dynamics lead to a linear time-invariant system
\begin{equation}\label{equ: global dynamics}
    \dot{\bm{z}} = 
-\qty(\bm{\Omega}\otimes\bm{I}_D)\bm{z}.
\end{equation}
It is straightforward to see from (\ref{equ: equilibrium config}) that $\bm{p}$ is an equilibrium of system (\ref{equ: global dynamics}), which is allowed by $\bm{\Omega}$ that has a $(D+1)$-dimensional nullspace. This is advantageous over the conventional average consensus system, where the graph Laplacian is used because the graph Laplacian has only a $1$-dimensional nullspace corresponding to the $\bm{1}$ vector, i.e., only configurations with all nodes in the same location can be in equilibrium.

Moreover, since $\bm{\Omega}$ is positive semidefinite, system (\ref{equ: global dynamics}) is globally and exponentially stable. We then claim that a formation defined by the framework $\mathcal{F} = (\mathcal{G},\bm{P})$ is stabilizable using system (\ref{equ: global dynamics}), i.e., the formation asymptotically converges to the solution space containing the equilibrium $\bm{p}$ given any initialization. Note that the targeted equilibrium $\bm{p}$ can be reached up to an affine transformation (hence its usefulness for AFC), but can also be uniquely determined given a few anchor nodes (which are commonly named leaders, e.g., in \cite{zhao2018affine}).

\subsection{Problem Formulation}
As discussed, the core element of the formation control system (\ref{equ: global dynamics}) is the stress matrix with the following properties:
\begin{enumerate}
    \item $\bm{\Omega}$ is positive semidefinite (PSD);
    \item $\bm{\Omega}$ has rank $N-D-1$;
    \item $\bm{\Omega}$ has rows of configuration $\bm{P}$ and $\bm{1}_N$ in the nullspace. 
\end{enumerate}
Given a configuration $\bm{P}$, we seek a sparse stress matrix satisfying the above properties representing a loosely connected graph, while maximizing the system's convergence rate. In consensus theory, the second smallest eigenvalue of the Laplacian governs convergence speed \cite{olfati2007consensus}, as it determines the slowest decaying component. Prior work has aimed to maximize this eigenvalue for faster convergence \cite{kim2005maximizing}. Similarly, for stress matrices, we aim to maximize the smallest nonzero eigenvalue, i.e., the ($D+2$)-th smallest eigenvalue of ${\bm{\Omega}}$ or $\lambda_{D+2}({\bm{\Omega}})$. We show later via numerical examples that this goal can be traded off for graph sparsification.

\section{Network Topology Sparsification}\label{sec: methodology}
In this section, we first propose a preliminary formulation directly related to the objectives and the constraints, which turns out to be a non-convex problem. We then focus on convexifying the formulation into a tractable convex optimization and give insights into the formulation. Recall from (\ref{equ: stress inc}) that the stress matrix can be constructed using the graph incidence matrix $\bm{B}$ given a stress vector $\bm{\omega}$. We initialize the graph as a complete graph with an incidence matrix $\bar{\bm{B}}\in\mathbb{R}^{N\times\bar{M}}$ where there are $\bar{M} = \frac{N(N-1)}{2}$ edges. We then aim to find a sparse vector $\bar{\bm{\omega}}\in\mathbb{R}^{\bar{M}}$ such that $\bm{\Omega}$ is sparse using (\ref{equ: stress inc}), but using the notation with bars. The effective number of edges from the sparse vector $\bm{\Omega}$ will be $M = \norm{\bar{\bm{\omega}}}_0$. Hence we pose the following problem $\mathcal{P}_0$:

\begin{subequations} \label{equ: primitive formulation}
  \begin{align}
  & \mathcal{P}_0: \quad \underset{\bar{\bm{\omega}},\bm{\Omega}}{\text{minimize}}
  && \norm{\bar{\bm{\omega}}}_0 - \alpha\lambda_{D+2}(\bm{\Omega}) \label{equ: obj pri} \\
  & \quad \quad \quad \text{subject to}
  && \bm{\Omega}=\bar{\bm{B}}\diag(\bar{\bm{\omega}})\bar{\bm{B}}\tp \label{equ: equality cstr pri}\\
  &&& \rank(\bm{\Omega}) = N-D-1 \label{equ: rank cstr pri}\\
  &&& \bm{\Omega}\succeq 0 \label{equ: PSD cstr pri}\\
  &&& \bm{\Omega}\bar{\bm{P}}\tp=\bm{0} \label{equ: null cstr pri}
  \end{align}
\end{subequations} where $\alpha$ is a weighting parameter, $\lambda_{D+2}(\bm{\Omega})$ is the smallest nonzero eigenvalue of $\bm{\Omega}$, and $\bar{\bm{P}} = [\bm{P}\tp, \bm{1}_N]\tp\in\mathbb{R}^{(D+1)\times N}$. As can be observed, $\mathcal{P}_0$ is difficult to solve due to the L0-norm and the rank constraint. In the next section, we convexify this problem to a tractable semidefinite program (SDP).

\subsection{Convexifying the Constraints}
Assuming a full-rank $\bar{\bm{P}}$, we let  $\bm{Q}\in\mathbb{R}^{N\times (N-D-1)}$ span the kernel of $\bar{\bm{P}}$ with orthonormal columns. Then, (\ref{equ: rank cstr pri}), (\ref{equ: PSD cstr pri}), and (\ref{equ: null cstr pri}) imply that $\bm{Q}\tp\bm{\Omega}\bm{Q}\in\mathbb{R}^{(N-D-1)\times(N-D-1)}$ is a positive definite (PD) matrix, since the null-subspace of $\bm{\Omega}$ is projected out by $\bm{Q}$, i.e.,  
$\bm{Q}\tp\bm{\Omega}\bm{Q}\succ 0$. Further substituting (\ref{equ: stress inc}), we have $\bm{\Psi}\diag(\bar{\bm{\omega}})\bm{\Psi}\tp\succ 0$, where $\bm{\Psi} = \bm{Q}\tp\bar{\bm{B}}$.


Substituting (\ref{equ: stress inc}) into  (\ref{equ: null cstr pri}), we have $\bar{\bm{P}}\bar{\bm{B}}\diag\qty(\bar{\bm{\omega}})\bar{\bm{B}}\tp = \bar{\bm{P}}\bar{\bm{B}}\diag\qty(\bar{\bm{\omega}})\qty[\bar{\bm{b}}_1,...,\bar{\bm{b}}_N] = \bm{0}$ where $\bar{\bm{b}}_i\in\mathbb{R}^{\bar{M}}, \forall i\in\mathcal{V}$ is the $i$-th column of $\bar{\bm{B}}\tp$. Observe that $\bar{\bm{P}}\bar{\bm{B}}\diag\qty(\bar{\bm{\omega}})\bar{\bm{b}}_i = \bar{\bm{P}}\bar{\bm{B}}\diag\qty(\bar{\bm{b}}_i)\bar{\bm{\omega}}=\bm{0}$. As such, we can construct a matrix $\bm{E}\in\mathbb{R}^{N(D+1)\times \bar{M}}$ with the structure
\begin{equation}\label{equ: Zhao2018E}
    \bm{E} = \mqty[\bar{\bm{P}}\bar{\bm{B}}\diag\qty(\bar{\bm{b}}_1) \\ \vdots \\ \bar{\bm{P}}\bar{\bm{B}}\diag\qty(\bar{\bm{b}}_N)],
\end{equation}
such that $\bm{E}\bar{\bm{\omega}} = \bm{0}$, which can replace constraint (\ref{equ: null cstr pri}).

\subsection{Eigenvalue Maximization}
We now give an explicit expression for the maximization of $\lambda_{D+2}(\bm{\Omega})$. Denoting $\bm{\psi} = \diag(\bm{\Psi}\tp\bm{\Psi})$, we show that maximizing $\bm{\psi}\tp\bar{\bm{\omega}}$ with a bounded $\norm{\bm{\Omega}}_2$ is equivalent to the eigenvalue maximization problem. To start, recollecting that $\bm{Q}\tp\bm{\Omega}\bm{Q} = \bm{\Psi}\diag(\bar{\bm{\omega}})\bm{\Psi}\tp$, we obtain
\begin{align}\label{equ: trace linear}
    \tr\qty(\bm{Q}\tp\bm{\Omega}\bm{Q}) 
  &=  \tr\qty(\bm{\Psi}\diag(\bar{\bm{\omega}})\bm{\Psi}\tp) = \tr\qty(\bm{\Psi}\tp\bm{\Psi}\diag(\bar{\bm{\omega}})) \\ \notag
  &= \tr\qty(\diag(\bm{\Psi}\tp\bm{\Psi})\bar{\bm{\omega}}) = \bm{\psi}\tp\bar{\bm{\omega}}.
\end{align}
Hence, the sum of the eigenvalues of the PD matrix $\bm{Q}\tp\bm{\Omega}\bm{Q}$, or equivalently the PSD matrix $\bm{\Omega}$ (since $\bm{Q}$ has orthonormal columns), can be represented by $\bm{\psi}\tp\bar{\bm{\omega}}$ where $\bm{\psi}$ is known. Then if we use an additional constraint $\norm{\bm{\Omega}}_2\leq\beta$ where we limit the largest eigenvalue of $\bm{\Omega}$ to $\beta>0$, $\lambda_{D+2}(\bm{\Omega})$ is maximized when $\lambda_{D+2}(\bm{\Omega})=...=\lambda_{\text{max}}(\bm{\Omega}) = \beta$ and the condition number $\kappa(\bm{\Omega})$ is minimized to $1$ at the same time. This eigenvalue maximization is guaranteed when there are no other objectives and constraints and will play as a relaxation under our proposed constrained formulation. It is worth mentioning that trace regularization is also commonly seen in network design literature such as \cite{kim2005maximizing}. Additionally, although we do not emphasize the minimization of the condition number in this work, it is shown in \cite{xiao2022framework} that such minimization can improve the robustness of AFC against time delays.


\subsection{Proposed Framework}
A common convex relaxation is that of the L0-norm to an L1-norm for sparsity, i.e., $\norm{\bar{\bm{\omega}}}_0$ to $\norm{\bar{\bm{\omega}}}_1$ in the objective. Finally, by combining all the convexified constraints and the objectives above, we present our proposed formulation $\mathcal{P}_1$
\begin{subequations} \label{equ: main formulation}
  \begin{align}
  & \mathcal{P}_1: \quad \underset{\bar{\bm{\omega}}}{\text{minimize}}
  && \norm{\bar{\bm{\omega}}}_1 - \alpha\bm{\psi}\tp\bar{\bm{\omega}} \label{equ: obj main} \\
  & \quad \quad \quad \text{subject to}
  &&\bm{\Psi}\diag(\bar{\bm{\omega}})\bm{\Psi}\tp - \gamma\bm{I}_{N-D-1} \succeq 0 \label{equ: PSD main}\\
  &&& \norm{\bar{\bm{B}}\diag(\bar{\bm{\omega}})\bar{\bm{B}}\tp}_2 \leq \beta \label{equ: 2norm main}\\
  &&& \bm{E}\bar{\bm{\omega}} = \bm{0} \label{equ: null main}
  \end{align}
\end{subequations}
where $\alpha>0$ is the weighting parameter of the objectives and $\beta>0$ upper bounds the spectral norm of $\bm{\Omega}$. We also set a lower bound $\gamma>0$ for (\ref{equ: PSD main}) for two explicit reasons, (a) a strict positive definite constraint $\bm{\Psi}\diag(\bar{\bm{\omega}})\bm{\Psi}\tp\succ 0$ is numerically infeasible, and (b) when the L1 term is dominating, we want to avoid obtaining the trivial solution $\bar{\bm{\omega}}=\bm{0}$. 

As a summary of the proposed formulation, with the help of (\ref{equ: null cstr pri}), constraints (\ref{equ: rank cstr pri}) and (\ref{equ: PSD cstr pri}) in $\mathcal{P}_0$ are convexified using (\ref{equ: PSD main}), and (\ref{equ: null cstr pri}) itself is rewritten with (\ref{equ: null main}) using only the variable $\bar{\bm{\omega}}$. Constraint (\ref{equ: 2norm main}) combined with the objective (\ref{equ: trace linear}) achieves the eigenvalue maximization. We now investigate the relationships among these hyperparameters in the following section.

\subsection{Choice of Hyperparameters}
The need for $\beta > \gamma$ is straightforward to ensure feasibility since they represent the largest and smallest eigenvalues, respectively. Additionally, $\beta$ and $\gamma$ should be sufficiently apart for a relatively large feasible region. The choice of $\alpha$ determines the objective function, which contains a sparsity term $\norm{\bar{\bm{\omega}}}_1$ (geometrically a cone) and a convergence term $-\alpha\bm{\psi}\tp\bar{\bm{\omega}}$ (a hyperplane). If the convergence term is dominant, i.e., $\alpha$ is large, the hyperplane will unfold the cone such that its global minimum vanishes, in which case the solution is minimized at the boundary of the feasible region which is not necessarily a sparse solution. Mathematically, to keep the global minimum of the objective function, $\bm{0}_{\bar{M}}$ should be in the subdifferential of the objective function, i.e.,
\begin{equation}
    \bm{0}_{\bar{M}} \in \partial\qty(\norm{\bar{\bm{\omega}}}_1 - \alpha\bm{\psi}\tp\bar{\bm{\omega}}),
\end{equation}
which directly yields that all elements of $\alpha\bm{\psi}$ are in the range $(-1,1)$. Recalling that all elements of $\bm{\psi}$ are nonnegative from the definition $\bm{\psi} = \diag(\bm{\Psi}\tp\bm{\Psi})$ and $\alpha >0$, we conclude that the choice of $\alpha$ that entails a sparse solution is
\begin{equation}\label{equ: range alpha}
    0 < \alpha < \frac{1}{\norm{\bm{\psi}}_\infty},
\end{equation}
where the infinity norm is the largest value of $\bm{\psi}$. Since a larger $\alpha$ promotes faster convergence, we recommend using the upper bound value of (\ref{equ: range alpha}) before fine-tuning. If $\alpha$ is chosen substantially beyond the range (\ref{equ: range alpha}), then the sparsity is traded off for the convergence rate, which we will illustrate in the next section.

\section{Simulations}\label{sec: simulations}
In this section, we validate our proposed framework through several examples, comparing its performance with state-of-the-art methods such as Lin et al. \cite{lin2015necessary}, Xiao et al. \cite{xiao2022framework}, and Yang et al. \cite{yang2018constructing}. We demonstrate the effect of varying the hyperparameter $\alpha$, with a fixed $\beta=1$ and $\gamma=0.1$. Since the approach in \cite{lin2015necessary} is not generative and requires a predefined graph, we first apply some sparse graph construction technique \cite{kelly2014class, connelly2022frameworks} 
to generate an appropriate structure before designing the stress, which is a widely adopted pipeline for AFC works. It is obvious from (\ref{equ: stress def}) and (\ref{equ: equilibrium config}) that stress has a scaling ambiguity, hence a scalar gain could amplify the smallest nonzero eigenvalue for an accelerated convergence. However, for a fair comparison of the convergence rate, we normalize all the acquired stress matrices for a maximum eigenvalue of $1$. The code for the simulations is available online \footnote{https://github.com/asil-lab/zli-sparse-urf}, in which we used CVXPY \cite{diamond2016cvxpy} to solve the proposed convex optimization.
\begin{figure*}[t]
	\centering	

        \subfloat[\scriptsize configuration]{\raisebox{0ex}
		{\includegraphics[width=0.14\textwidth]{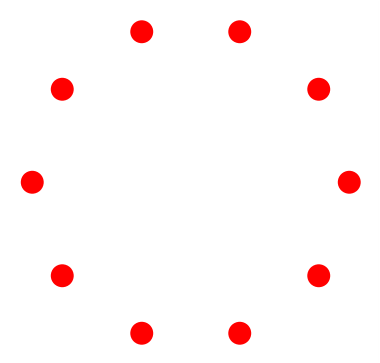}}%
	}
        \hspace{1mm}
	\subfloat[\scriptsize Yang et al.]{\raisebox{0ex}
		{\includegraphics[width=0.14\textwidth]{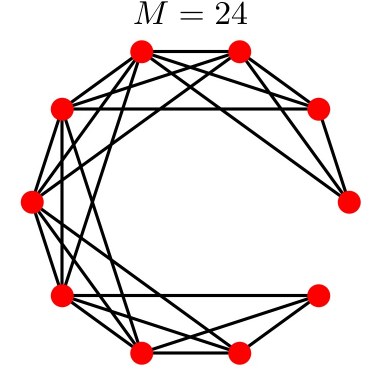}}%
	}
    \hspace{1mm}
        \subfloat[\scriptsize 
        Lin et al.]{\raisebox{0ex}
    		{\includegraphics[width=0.14\textwidth]{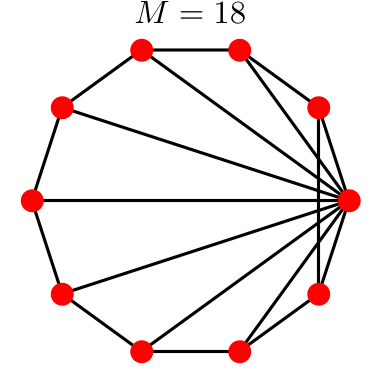}}%
    	}
        \hspace{1mm}
        \subfloat[\scriptsize proposed, $\alpha=0.5$]{\raisebox{0ex}
    		{\includegraphics[width=0.14\textwidth]{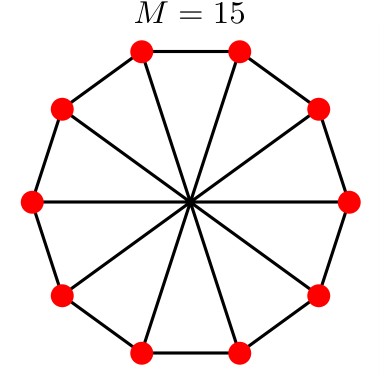}}%
    	}
        \hspace{1mm}
        \subfloat[\scriptsize proposed, $\alpha=1.5$]{\raisebox{0ex}
    		{\includegraphics[width=0.14\textwidth]{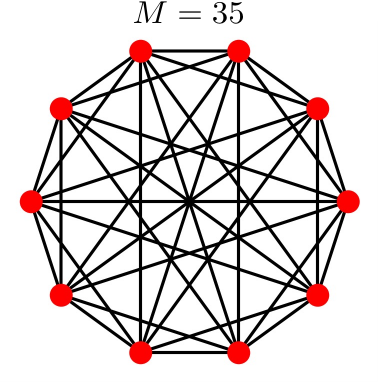}}%
    	}
        \hspace{1mm}
        \subfloat[\scriptsize proposed, $\alpha=5$]{\raisebox{0ex}
    		{\includegraphics[width=0.14\textwidth]{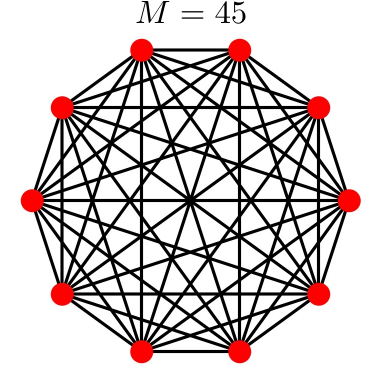}}%
    	}

	\caption{Resulting graphs of our proposed solutions compared with the state-of-the-art. A 10-node regular polygon has a maximum of $\bar{M} = 45$ edges. The recommended value of $\alpha$ for a sparse solution is $1/\norm{\bm{\psi}}_\infty = 0.51$. 
    The structure from \cite{kelly2014class} is used for
    Lin et al.}
	\label{fig: sparse circular}
		
\end{figure*}
Our main test case is a nongeneric regular polygon with 10 nodes in $\mathbb{R}^2$ shown in Fig. \ref{fig: sparse circular}(a), a common geometric pattern in, e.g., antenna arrays for maximizing aperture. The optimized graphs are shown in Fig. \ref{fig: sparse circular} where all methods yield a sparsified solution compared with a complete graph, among which our proposed solution offers the least required edges under the right $\alpha$ setting (\ref{equ: range alpha}). Also, the resulting graph becomes denser with a larger $\alpha$, which substantiates our previous discussion that the optimal solution will fall on the border of the feasible region. 
Besides, it is also consistent with the intuition that the sparsity term in (\ref{equ: obj main}) is less dominant with a larger $\alpha$, so a less sparse solution is achieved. Another observation from the results is that the proposed framework gives symmetric solutions as compared to the state-of-the-art works, which is a key benefit for robustness and for balancing the network load. This is because the proposed framework does not assume any sparsity pattern like \cite{yang2018constructing} or \cite{kelly2014class}, although symmetry is not explicitly enforced. 

\begin{figure}[t]
    \centering%
    \includegraphics[width=0.4\linewidth]{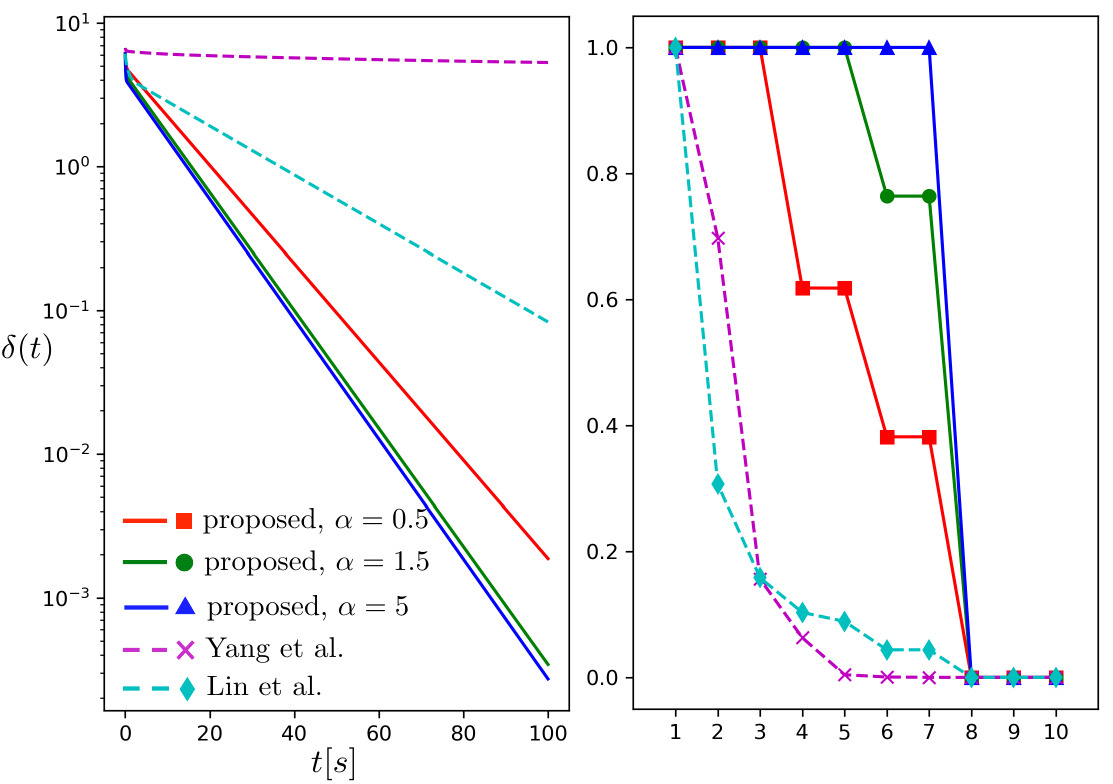}
    \caption{The convergence of the formation control system (left). The tracking error is $\delta(t) = \norm{\bm{z}(t) - \bm{p}}_2$. The normalized eigenvalues of the stress matrices (right). The 7th eigenvalue corresponds to the error convergence speed. Lin et al. is designed over graph Fig. \ref{fig: sparse circular} (c).}%
    \label{fig: circular eigs and cvg}%
\end{figure}

The convergence of our proposed solution is shown in Fig.~\ref{fig: circular eigs and cvg}. We compare the eigenvalues of the acquired stress matrix with the resulting convergence of formation control. When $\alpha$ is tuned to a larger value, the smallest nonzero eigenvalue and the convergence speed are promoted. The extreme case is $\alpha=5$, where all nonzero eigenvalues are maximized to $1$ at the cost of a full connection among nodes in Fig. \ref{fig: sparse circular}(f). This confirms that the denser the connections, the more efficiently information is spread in the network hence a faster convergence. It is worth noting that our most sparse solution, i.e., $\alpha=0.5$, still outperforms the existing works.

The generic example from \cite{xiao2022framework} is shown in Fig. \ref{fig: sparse xiao} where there are 6 nodes in $\mathbb{R}^2$. All compared methods present a reduced number of edges from a maximum of $\bar{M} = 15$ and the proposed technique can reach the lower bound $\underline{M}=2N-2 = 10$ for generic configurations \cite{kelly2014class}. Additionally, the proposed solution entails the highest $\lambda_4(\bm{\Omega})$, which is the smallest nonzero eigenvalue.

\section{Conclusion}\label{sec: conclusions}
In this work, we formulated the stress matrix design for a universally rigid framework as a convex optimization problem, aiming to minimize the number of edges required for network connectivity while enhancing convergence speed. We provided insights into the selection of hyperparameters, and numerical results demonstrated that our method outperforms existing approaches in both sparsity and convergence speed across our test cases. In future work, we plan to benchmark our solution across broader scenarios with more evaluation criteria such as control energy. We also aim to incorporate additional constraints, such as distance constraints and balanced load distribution, to refine our network design further.

\begin{figure}[t]
    \centering
    \subfloat[\scriptsize The resulting graphs.]{\raisebox{1mm}{ 
        \includegraphics[width=0.25\textwidth]{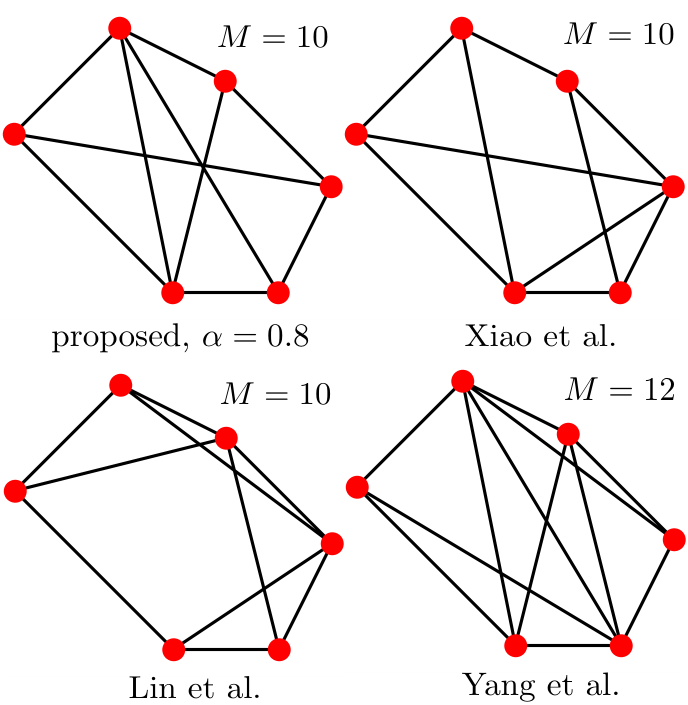}
    }}
    \hspace{0pt} 
    \subfloat[\scriptsize Normalized eigenvalues of the acquired stress matrices.]{
        \includegraphics[width=0.18\textwidth]{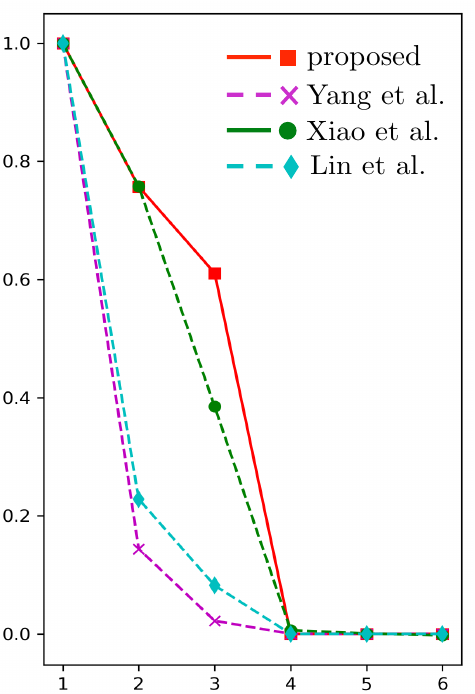}
    }
    
    \caption{Comparison of results using the generic configuration from \cite{xiao2022framework}. A structure from \cite{connelly2022frameworks} is used for Lin et al.}
    \label{fig: sparse xiao}
\end{figure}
\bibliographystyle{IEEEtran}
\bibliography{refs}

\end{document}